\documentclass[runningheads,a4paper]{llncs}
\usepackage{graphicx}
\usepackage{latexsym,amssymb,amsmath}
\usepackage[utf8]{inputenc}
\usepackage[all]{xy}
\usepackage{gastex}
\usepackage{color}
\usepackage[active]{srcltx}

\label{definitions}

\usepackage{url}
\urldef{\mailsa}\path|{klima,polak}@math.muni.cz|

\begin{document}
\title{Syntactic Structures of Regular Languages}
\author{Ond\v rej Kl\'{i}ma and Libor Pol\'ak
\thanks{The authors were supported by the 
Institute for Theoretical Computer Science (GAP202/12/G061),
Czech Science Foundation.}}
\institute{Department of Mathematics and Statistics, Masaryk University\\
Kotl\'a\v rsk\'a 2, 611 37 Brno, Czech Republic\\
\mailsa\\
\url{http://www.math.muni.cz}}

\maketitle

\begin{abstract}
We introduce here the notion of syntactic lattice algebra
which is an analogy of the syntactic monoid and of the syntactic
semiring. We present a unified approach to get those three structures.
\end{abstract}

\section{Introduction}\label{Introduction}

The algebraic theory of regular languages arose with 
the Eilenberg theorem~\cite{e} which establishes  bijection between
the class of all varieties of regular languages and
the class of all pseudovarieties of finite monoids.
In this correspondence a given language belongs to a certain variety of regular languages if and only if the syntactic monoid
of the language belongs to the corresponding pseudovariety of finite monoids.
The original motivation of that theory was looking for algorithmic procedures for deciding the membership 
in various significant classes of regular languages.
From that reason the equational description of pseudovarieties of finite algebras by Reiterman~\cite{reiterman}
plays a useful role for deciding the membership in corresponding pseudovarieties of monoids.

Recall that a variety of regular languages is a class closed under 
Boolean operations, quotients and preimages in homomorphisms. 
Since  not all natural classes of regular languages are varieties,
the research in algebraic theory of regular languages
was later also devoted to generalizations of the Eilenberg correspondence 
to a more general kind of classes.
The first such contribution~\cite{pin-eilenberg} introduced positive varieties 
of languages for which pseudovarieties of finite ordered monoids 
are appropriate algebraic counterparts. Here a positive variety need not to be closed under complementation
and the correspondence uses the fact that the syntactic monoid of a regular language is implicitly 
equipped by the compatible partial order.
Then the second author in~\cite{lp-examples} introduced the notion of disjunctive varieties of regular languages 
which need not to be closed under intersection (and complementation) and for which pseudovarieties of 
finite idempotent semirings were considered. 
Now one uses the syntactic semiring of a language  in the correspondence,  when he/she wants
to test the membership of that language in a considered disjunctive variety.

On the other hand Straubing~\cite{s} introduced the notion of 
$\mathcal C$-varieties which need not to be 
closed under preimages in all homomorphism, but only those from the specific class of homomorphism $\mathcal C$.  
Note that this concept can be combined with the mentioned variants of the Eilenberg correspondence,
where positive $\mathcal C$-varieties and disjunctive $\mathcal C$-varieties are already 
considered in~\cite{chaubard} and~\cite{lp-varieties}, respectively. 
Now the syntactic structure is the whole syntactic homomorphism.
Another generalization was done in~\cite{ggp-icalp-2008} 
where a single alphabet is considered. Here Boolean algebras and lattices of regular languages 
(over a fixed alphabet) are 
studied  on the side of classes of regular languages.
All variants of the Eilenberg theorem proved their usefulness due to  
existing characterizations via equational  descriptions for the corresponding pseudovarieties 
of the syntactic structures --- relevant references are~\cite{reiterman}
for the classical Eilenberg correspondence, \cite{pin-weil} for pseudovarieties of ordered monoids, 
\cite{kunc} for $\mathcal C$-pseudovarieties, and~\cite{ggp-icalp-2008} for Boolean algebras and lattices. 

Certain modifications of Eilenberg theorem outside regular languages based on Stone duality 
from~\cite{ggp-icalp-2008} were developed in last five years --- see e.g.~\cite{boj,borlido}.
Further, some papers started to analyze a categorical generalizations of the Eilenberg theorem.
In particular~\cite{uramoto} introduced so-called semi-galois categories and the Eilenberg theorem for them, 
while Ad\'amek and et. in the series of conference papers~\cite{adamek-1,adamek-2}
studied certain pairs of dual monoidal categories of (ordered) algebras. 
Consequently, a uniform description of what is the Eilenberg theorem  was described in~\cite{adamek-3}. 
All this categorical work is put together in the long paper~\cite{adamek-4}[version 3] 
however some other developments can be expected in near future.
Notice that one of the statements from~\cite{adamek-4} is that all 
mentioned Eilenberg correspondences can be obtained as an application 
of their main general categorical version of Eilenberg type theorem. 
In particular, the case of disjunctive varieties corresponding to 
pseudovarieties of idempotent semirings.

The aim of the present contribution is to introduce modification
of the notion of the syntactic monoid which would be useful in other variants of Eilenberg type theorems.
The class of languages which we would like to consider are not closed under
any Boolean operation. Therefore, the work can be viewed as a continuation of the work concerning 
disjunctive varieties of languages.
The experience with Eilenberg type correspondences gives us an intuitive idea 
that 
when one looses closure properties on the side of classes
of languages, then one needs to consider a richer syntactic structure of the 
language.
In this research we try to prolong this naive idea in such a way that we try to complete the syntactic semiring 
into a distributive lattice. 
Unfortunately, a potential Eilenberg type theorem does not follow from the mentioned general categorical results.
The problem is that the category of distributive lattices does not 
satisfied assumptions specified in papers~\cite{adamek-1,adamek-2,adamek-3,adamek-4}.\footnote{As 
mentioned above, the category of idempotent semirings does.}
In particular, if one takes for the category $\mathcal D$ in~\cite{adamek-1} bounded distributive lattices, 
then the category $\mathcal C$ would be the category
of ordered sets, which is not considered in these papers -- the ordered algebras are considered 
only in $\mathcal D$. Another and probably more significant difference is
that bounded distributive lattices do not satisfy the assumption 4.8.(b) in~\cite{adamek-1}, 
since homomorphisms between distributive lattices are not naturally
equipped by the structure of distributive lattices, 
which seems to be an essential condition in the mentioned categorical approach.
Notice also that the theory of semi-galois categories from~\cite{uramoto} can 
not be also applied, 
since the basic assumption  is that pushouts and pullbacks in the semi-galois category
need to be computed in the same way as in the category of sets $\mathbf{Set}$. 
This is not the case for the category of idempotent semirings neither the category of bounded distributive lattices.  
From all that reasons, we strongly believe that the study of the uniform  approach to syntactic structures presented 
in this paper could lead to a new type of Eilenberg correspondences, even if the techniques in the paper 
are quite elementary. 

\bigskip

The basic approach of the present paper can be briefly explained in the following way. 
It is well-known that the syntactic monoid of a language $L$  over the alphabet $A$
can be viewed as the transformation monoid 
of the minimal complete deterministic automaton $\mathsf D_L$ of $L$.
More precisely, we let words of $A^*$ act on states of
$\mathsf D_L$ and the composition of such transformations 
corresponds to multiplication in the syntactic monoid.

By Brzozowski construction each state of the minimal automata can be identified with
the language accepted from that state, therefore the elements of $A^*$ 
can be considered as unary operations 
on the set $2^{A^*}$ of all languages over $A$. 
These unary operations are compositions of basic unary operations given 
by letters. 
Since the composition of mappings is associative, 
compositions of unary operations correspond exactly to words.
To get analogues of the monoid $A^*$, we consider structures 
with more operations, namely we  
use here the following three term algebras:
\begin{itemize}
\item
$F$ is the absolutely free algebra over the alphabet $A$ 
with the operation symbol $\cdot$ and nullary symbol $\lambda$,
\item
to get $F'$ we enrich the previous signature by binary $\wedge$ and nullary
$\top$,
\item
to get $F''$ we enrich the last signature by binary $\vee$ and nullary
$\perp$.
\end{itemize}
Now we let 
our terms act 
on the set $2^{A^*}$ in a natural way 
(the formal definitions are in Section~\ref{s:transformation} ).
We show that identifying terms of $F$ ($F'$ and $F''$) 
giving the same transformations, we get
exactly the free monoid $A^*$ over $A$, (the free semiring $A^\square$
over $A$ and the free, so-called, lattice algebra $A^\diamond$
over $A$, respectively).
Let us stress that all our considerations concern three levels: level of monoids -- 
the classical one (Pin~\cite{pi,pin-varieties}), 
level of semirings  
(considered also in Pol\'ak~\cite{lp-examples,lp-syntactic}), 
and that of lattice algebras -- a new contribution.

When generating subalgebras in $2^{A^*}$ 
by a single regular language $L$
using terms from $F$, $F'$ and $F''$,
and choosing the final states
appropriately, we get the classical minimal complete deterministic
finite automaton of $L$ (here called the canonical finite automaton
of $L$), the canonical meet automaton of $L$ (see Section~6 of 
Pol\'ak~\cite{lp-examples}) and the canonical
lattice automaton of $L$, respectively.
Section~\ref{s:canonical}
is devoted to canonical automata,
then transforming those automata accordingly, we get the 
corresponding syntactic structures
in Section~\ref{s:syntactic} in all three levels.
Our constructions are also accompanied by examples.
Moreover, a significant instance of a future
Eilenberg type theorems is presented.

\section{Specific Algebraic Structures}\label{s:specific-algebraic}

Usually, a semiring has two binary operations denoted by $+$ and $\cdot$, where the 
neutral element for $+$ is denoted by $0$.
Since we work with idempotent semirings, which can be naturally ordered, we use the symbol $\wedge$ instead of $+$, and 
the symbol $\top$ instead of $0$  in the following basic definition. 
By an {\em idempotent semiring} we mean the structure 
$(S,\wedge,\cdot,\top,1)$ where 
$(S,\wedge,\top)$ is a commutative idempotent monoid, also called {\it semilattice},  with the neutral element $\top$,
$(S,\cdot,1)$ is a monoid with the neutral element 1 and the zero element $\top$, 
and
the operations $\wedge$ and $\cdot$ satisfy the usual distributivity laws 
$$(\, \forall\, a,b,c \in S\,)\quad  a\cdot(b\wedge c)=a\cdot b \wedge a\cdot c, \ \ (b\wedge c)\cdot a=b\cdot a \wedge c \cdot a \, .$$
The set $S$ can be naturally ordered: for every $a,b\in S$
we have $a\leq b$ if and only if $a\wedge b=a$. Then $\top$ becomes the greatest element in $(S,\leq)$. 
This explains our choice of the symbol $\top$.

The elements of the free idempotent semiring $A^\square$
over the set $A$
can be represented
by finite subsets of $A^*$. 
This representation is one-to-one.
Operations are the operation of union and the obvious 
multiplication, $\emptyset$
is the neutral element for $\wedge$, it is the zero for $\cdot$ and 
$\{\lambda\}$ is the neutral element for the multiplication. 
If we identify each word $u\in A^*$ with the element $\{u\}\in A^\square$, then we can see
$A^*$ as a subset of $A^\square$. Under this identification, for each 
$U=\{u_1,\dots,u_k\},\, k> 0,\, u_1,\dots,u_k\in A^*$
we can write $U=u_1\wedge\dots\wedge u_k$.

A next structure we use 
is the free bounded distributive
lattice $A^\diamond$ over $A^*$. 
The representation of the free bounded distributive lattice $F_{BDL}(X)$ over a finite set $X$ is well-known
(see e.g. Gr\"atzer~\cite{gratzer}). Usually the elements of $F_{BDL}(X)$ are represented by upper sets in 
$(\mathcal P (X),\subseteq)$ (here $\mathcal P(X)$ denotes the set of all 
subsets of $X$)
with the operation intersection and union. Moreover $\emptyset$ is the 
smallest element, i.e. 
$\emptyset=\bot$  and $\mathcal P(X)$ is the greatest element, i.e $\mathcal P(X)=\top$. 
Alternatively, each such upper set can be represented just by its minimal elements -- then 
$\bot$ is still represented by $\emptyset$, however $\top$ is now 
represented by $\{\emptyset\}$. 
Both these representations work with terms in the form 
\begin{equation}\label{e:new-free-BDL}
(x_{1,1}\wedge\dots\wedge x_{1,r_1})\vee\dots\vee
 (x_{k,1}\wedge\dots \wedge x_{k,r_k})\, .
 \end{equation}
The first representation adds as much as possible conjunctions into the 
form~(\ref{e:new-free-BDL}), however
the second representation in contrary remove all superfluous conjunctions to get the shortest expression as possible. 

The free bounded distributive lattice over a countable set $X$ can be obtained as a union of 
the free bounded distributive lattices over finite subsets of $X$.  
Here we just describe the resulting structure in the case when
$X$ is  equal to $A^*$ where we use the second representation from the previous paragraph.
The elements of $A^\diamond$
are of the form
\begin{eqnarray}\label{e:incomparable-U}
& \{\{u_{1,1},\dots,u_{1,r_1}\},\dots,\{u_{k,1},\dots,
 u_{k,r_k}\}\},\ \text{where } k,r_1,\dots,r_k\geq 0, 
 & \nonumber \\
&  u_{i,j}\in A^*\  \text{for } 
i=1,\dots,k,\ j=1,\dots,r_i\,  
 \text{and the inner sets}  & \\ 
& \{u_{i,1},\dots,u_{i,r_i}\}\text{'s 
are incomparable with respect to} \subseteq \, . & \nonumber
\end{eqnarray}
The interpretation of the element of the form (\ref{e:incomparable-U}) is 
\begin{equation}\label{e:interpretation}
(u_{1,1}\wedge\dots\wedge u_{1,r_1})\vee\dots\vee
 (u_{k,1}\wedge\dots \wedge u_{k,r_k})\, .
 \end{equation}
Particularly, each element of the form $\{\{u_1,\dots,u_{k}\}\}\in A^\diamond$
is identified with $u_1\wedge  \dots\wedge u_{k}$,
which is equal to $\{u_1,\dots,u_{k}\}$ in  $A^\square$.
Thus we can see $A^\square$ as a subset of  $A^\diamond$  under the identification $U\mapsto \{U\}$.
Defining the operations $\wedge$ and $\vee$ on the set $A^\diamond$,
one uses the form (\ref{e:interpretation}) for the element of the form (\ref{e:incomparable-U}). 
In the case of the definition of $\vee$, one omits 
the superfluous $(u_{j,1}\wedge\dots \wedge u_{j,r_j})$'s, while in the case
of the definition $\wedge$ one uses the distributivity law first 
and then again omits the superfluous $(u_{j,1}\wedge\dots \wedge u_{j,r_j})$'s. 
In this way one gets the (unique) element of the form (\ref{e:incomparable-U}) in both cases.
Notice that $\{\emptyset\}$ is the greatest element in $A^\diamond$ and $\emptyset$ is the smallest one.

We equip the structure $A^\diamond$  with a multiplication, namely we extend 
the multiplication from $A^*$ to 
 $A^\diamond$ using
 \begin{eqnarray} 
&{\mathcal U}\cdot({\mathcal V}\wedge{\mathcal W})= {\mathcal U}\cdot {\mathcal V} \wedge {\mathcal U}\cdot {\mathcal W},\ 
 ({\mathcal U}\wedge {\mathcal V})\cdot w= 
 {\mathcal U}\cdot w \wedge {\mathcal V}\cdot w\, ,& \label{e_nasobeni} \\
 &{\mathcal U}\cdot({\mathcal V}\vee{\mathcal W})= {\mathcal U}\cdot {\mathcal V} \vee {\mathcal U}\cdot {\mathcal W},\ 
 ({\mathcal U}\vee {\mathcal V})\cdot w= 
 {\mathcal U}\cdot w \vee {\mathcal V}\cdot w\, ,&\nonumber 
 \end{eqnarray}
 for ${\mathcal U},{\mathcal V},{\mathcal W}\in A^\diamond,\,w\in A^*$.

 \medskip
In this paper we consider
various kinds of automata. All of them 
are deterministic and complete, however they could have 
an infinite number of states.
When using the term  {\it semiautomata}, no initial
nor final states are specified.

Having an equivalence relation $\rho$ on a set $G$ and an element
$a\in G$, we denote by $a\rho$ the class of $G/\rho$ containing $a$.

\section{Transformation structures}\label{s:transformation}

Let $A$ be a finite non-empty set. The aim of this section is to elaborate  
the actions of term algebras mentioned above on languages over the alphabet $A$.

For $u\in A^*$ and $L\subseteq A^*$, we write $u^{-1}L=\{\, v\in A^* \mid uv\in L\, \}$.
We speak about a~{\em left quotient} of $L$.
\medskip
 
\noindent {\bf Monoids.} Let $F$ be the 
absolutely free algebra (that is, the algebra of all terms)
 over a set $A$ with respect to the
 binary operational symbol $\cdot$ and
nullary operational symbol $\lambda$.

We define inductively the actions of elements of $F$ on subsets of
 ${A^*}$ :
\begin{equation}
 L\circ \lambda = L,\,\  L\circ a=a^{-1}L \text{ for } a\in A,\, 
 \ L\circ (u \cdot v)= 
 (L\circ u)\circ v \text{ for } u,v\in F
 \,. \label{e-monoidy}
 \end{equation}
 
 This leads to a natural identification of certain pairs 
 of elements of $F$, namely:
 for $u,v\in F$, we put
 $u\mathrel{\rho^*} v$ if and only if $(\,\forall\ L\subseteq A^* )\ L\circ u= L\circ v $.
 \def\m{\mathrel{\rho^*}}

\begin{proposition}\label{p:free-monoid}
 The relation $\rho^*$ is a congruence relation on $F$ and
 $F/\rho^*$ is isomorphic to the free 
 monoid $A^*$ over $A$ via the extension
 of the mapping $a\rho^* \mapsto a,\ a\in A$.
\end{proposition}

\begin{proof}
Let $u,v,w\in F$. 
If $u \m v$ then, for each $L\subseteq A^*$, we have
$L\circ u=L\circ v$. Therefore 
$L\circ (u\cdot w) = (L\circ u)\circ w = 
(L\circ v)\circ w = L\circ (v\cdot w)$, which gives 
$u\cdot w \m v\cdot w$.
Similarly,
$L\circ (w\cdot u) = (L\circ w)\circ u = 
(L\circ w)\circ v = L\circ (w\cdot v)$, which gives 
$w\cdot u \m w\cdot v$.
Thus $\m$ is a congruence relation on $F$.

Now we prove that, for each $u,v,w\in F$, we have
$(u\cdot v) \cdot w \m u\cdot(v\cdot w)$, $\lambda\cdot u \m u$ and
$u\cdot \lambda \m u$.
Indeed, choosing $L\subseteq A^*$, it holds
$L\circ((u\cdot v)\cdot w)=(L\circ(u\cdot v))\circ w
=((L\circ u)\circ v)\circ w=(L\circ u)\circ (v\cdot w)
=L\circ (u\cdot (v\cdot w))$.
Furthermore,
$L\circ (\lambda\cdot u)=(L\circ\lambda)\circ u= L\circ u,
$ and $L\circ (u\cdot \lambda)=(L\circ u)\circ\lambda=L\circ u$.

Thus we can omit brackets in elements of $F$ and $\lambda$ acts as a neutral
element. Therefore every element of $F/\rho^*$ can be represented 
by a word from $A^*$.
It remains to show that different words $u$ and $v$ represent different
elements of $F/\rho^*$. Indeed,
for $u\not =v$, we have $\lambda\in \{u\}\circ u$ but $\lambda\not\in \{u\}\circ v$.
\qed\end{proof}

\medskip
 
\noindent {\bf Semirings.} Let $F'$ be the absolutely free algebra 
 over $A$ with respect to the
 operational symbols $\cdot,\lambda$, binary symbol $\wedge$ 
 and nullary symbol $\top$. 
 We define inductively the actions of elements of $F'$ on $2^{A^*}$:
 we use the formulas from (\ref{e-monoidy}) for $u,v\in F'$
 and
 \begin{equation}
  L\circ \top = A^*,\ L\circ (u \wedge v)= (L\circ u)
  \cap (L\circ v) \text{ for } u,v\in F'\,.
  \label{e-semiringy}
  \end{equation}

  Again, it leads to certain identification of pairs of elements of
  $F'$, namely:
for $u,v\in F'$, we put $u\mathrel{\rho^\square} v$ if and only if  $(\,\forall\ L\subseteq A^* )\ 
L\circ u= L\circ v$.
\def\s{\mathrel{\rho^\square}}
\def\l{\mathrel{\rho^\diamond}}
\def\b{\mathrel{\rho^\circ}}

\def\ml{\mathrel{\sim_L^*}}
\def\sl{\mathrel{\sim_L^\square}}
\def\ll{\mathrel{\sim_L^\diamond}}

\begin{proposition}\label{l:free semiring}
The relation $\rho^\square$ 
 is a congruence relation on $F'$ and
 $F'/\rho^\square$ is isomorphic to the 
 free idempotent semiring 
 $A^\square$ over $A$ via the extension
 of the mapping $a\rho^\square \mapsto a,\ a\in A$.
\end{proposition}

\begin{proof}
Let $u,v,w\in F'$. 
If $u \s v$ then, for each $L\subseteq A^*$, we have
$L\circ u=L\circ v$. 
We get $u\cdot w \s v\cdot w$
and $w\cdot u \s w\cdot v$ as in the case of
Proposition~\ref{p:free-monoid}.

Furthermore, 
$L\circ (u\wedge w) = (L\circ u)\cap (L \circ w) = 
(L\circ v)\cap (L\circ w) = L\circ (v\wedge w)$, 
which gives $u\wedge w \s v\wedge w$.
In the same way we can prove that
$w\wedge u \s w\wedge v$.
\ Thus $\s$ is a congruence relation on $F'$.

Now we show that $(F'/\rho^\square,\wedge,\top\rho^\square)$ is a commutative 
idempotent monoid with  the neutral element $\top\rho^\square$.
The commutativity and associativity of $\wedge$ is clear
as well as the fact that $\top\rho^\square$ is a neutral element for 
the operation $\wedge$. To show the idempotency of $\wedge$
notice that, for each $L\subseteq A^*$ and $u\in F'$, we have $L\circ(u\wedge u)
=(L\circ u)\cap (L\circ u)=L\circ u$.

The proof of the associativity of $\cdot$ on $F'/\rho^\square$ 
and the fact that $\lambda\rho^\square $ is a neutral element for 
the operation $\cdot$ is similar to that for monoids.
The fact that $\top\rho^\square$ is a zero element for $\cdot$ is clear.

Finally, we prove the distributivity laws.
Let $L\subseteq A^*$, $u,v,w\in F'$. Then
$L\circ (u\cdot(v\wedge w))
=(L\circ u)\circ(v\wedge w)
=(L\circ u)\circ v \cap (L\circ u)\circ w
=L\circ u\cdot v \cap L\circ u\cdot w
=L\circ(u\cdot v \wedge u\cdot w)$.
Similarly,
$L\circ((u\wedge v)\cdot w)
=(L\circ(u\wedge v))\circ w
=(L\circ u \cap L\circ v)\circ w
=(L\circ u)\circ w \cap (L\circ v)\circ w
=(L \circ u\cdot w)\cap (L\circ v\cdot w)
=L\circ (u\cdot w \wedge v\cdot w).$

We have proved that 
$F'/\rho^\square$ with the appropriate operations
is an idempotent semiring.
Therefore every element of $F'/\rho^\square$ can be represented 
by $u_1\wedge \dots \wedge u_k$ with $k\geq 0$
and $u_1,\dots,u_k\in A^*$.  
To get the unique representation of such element we use the 
idempotency and commutativity law and represent the element in
$F'/\rho^\square$ by the set 
$\{u_1,\dots,u_k\}$.
Having such two different sets 
$\{u_1,\dots,u_k\}$ and $\{v_1,\dots,v_\ell\}$, $\ell\geq 0,\, 
v_1,\dots,v_\ell\in A^*$, we show that 
they are not $\s$-related.
Indeed, put $L=\{u_1,\dots,u_k\}$. Then 
$\lambda\in L\circ \{u_1,\dots,u_k\} 
= u_1^{-1} L \cap \dots \cap u_k^{-1} L$ and
$\lambda\in L\circ \{v_1,\dots,v_\ell\}$ would give
$\{v_1,\dots,v_\ell\} \subsetneqq \{u_1,\dots,u_k\}$.
Take $L=\{v_1,\dots,v_\ell\}$ in this case.
\qed\end{proof}

\medskip
  
\noindent {\bf Lattice algebras}. Let $F''$ be the absolutely free 
algebra over $A$ with respect to the
 operational symbols $\cdot,\lambda,\, \wedge$, $\top$,
 binary $\vee$ and nullary $\perp$. 
 We use (\ref{e-monoidy}), (\ref{e-semiringy}) with
 $u,v\in F''$ 
 and
 \begin{equation}
  L \circ \perp=\emptyset,\ 
  L\circ (u \vee v)= (L\circ u) \cup (L\circ v) \text{ for } 
  u,v\in F''\ .
  \label{e-svazy}
  \end{equation}
  Again, it leads to certain identification of pairs of elements of
  $F''$, namely:
  for $u,v\in F''$, we put $u\mathrel{\rho^\diamond} v$ if and only if 
$(\ \forall\,  L\subseteq A^*) 
 \ L\circ u= L\circ v$.
 
\begin{proposition}\label{l:free lattice algebras}
The relation $\rho^\diamond$ 
 is a congruence relation on $F''$ and $F''/\rho^\diamond$ is isomorphic 
 to the free bounded distributive 
 lattice $A^\diamond$ over $A^*$ equipped 
 with multiplication satisfying 
 (\ref{e_nasobeni}), 
 via the extension of the mapping 
 $a\rho^\diamond \mapsto a,\ a\in A$.
\end{proposition}

\begin{proof}
Let $u,v,w\in F''$. 
If $u \l v$ then, for each $L\subseteq A^*$, we have
$L\circ u=L\circ v$. 
We get $u\cdot w \l v\cdot w$,
$w\cdot u \l w\cdot v$, 
$u\wedge w \l v\wedge w$,
$w\wedge u \l w\wedge v$
as in the case of Proposition~\ref{l:free semiring}.
Furthermore, 
$L\circ (u\vee w) = (L\circ u)\cup (L \circ w) = 
(L\circ v)\cup (L\circ w) = L\circ (v\vee w)$, 
which gives $u\vee w \l v\vee w$.
In the same way we can prove that
$w\vee u \l w\vee v$.
Thus $\l$ is a congruence relation on $F''$.

Now we state the properties of operations $\wedge,\vee,\cdot,
\top, \perp$ and $\lambda$ on $F''/\rho^\diamond$.
Proofs of all statements are straightforward and therefore omitted.
The operation $\wedge$ is commutative, associative and idempotent,
$\top$ is the neutral element and $\perp$ is the zero.
The operation $\vee$ is commutative, associative and idempotent,
$\perp$ is the neutral element and $\top$ is the zero.
The operations $\wedge$ and $\vee$ are connected by the
distributivity laws.
The operation $\cdot$ is associative, $\lambda$ is 
the neutral element, $\top$ and $\perp$ are right zeros, 
and $\top\cdot a \l \top$, $\perp\cdot a \l \perp$
for all $a\in A$.
Finally, the distributivity $u\cdot (v \vee w)=u\cdot v \vee u\cdot w$
holds for arbitrary $u,v,w\in F''$
and the distributivity $(u\vee v)\cdot w=u\cdot w \vee v\cdot w$ for 
$u,v\in F''$ and $w\in A^*$.
Similarly for the operation $\wedge$.

We have proved that every element of $F''/\rho^\diamond$ can be 
represented 
as $$(u_{1,1}\wedge\dots\wedge u_{1,r_1})\vee\dots\vee
 (u_{k,1}\wedge\dots \wedge u_{k,r_k})\, ,$$ where
$k,r_1,\dots ,r_k \ge 0$ and 
$u_{i,j}\in A^*$ for all $i=1,\dots,k,\ j=1,\dots,r_i$.
(Here $k=0$ corresponds to the element $\perp$
and $k=1$, $r_1=0$ corresponds to the element $\top$.)
Using the idempotency and commutativity of $\wedge$ and $\vee$
we can write such element even as
$\{\{u_{1,1},\dots,u_{1,r_1}\},\dots,\{u_{k,1},\dots,
 u_{k,r_k}\}\}$.
To get canonical forms remove the richer one from each pair
of comparable inner sets.

Let $\mathcal U=\{U_1,\dots,U_k\}$ and $\mathcal V=\{V_1,\dots,V_\ell\}$
be different canonical forms. 
We show that ${\mathcal U}$ and ${\mathcal V}$ represent elements of $F''$ which are not
$\l$-related. If $U_i\not\in {\mathcal V}$, take $L=U_i$. Then $\lambda\in L\circ {\mathcal U}$
and $\lambda\in L\circ {\mathcal V}$ would give that $V_j\subseteq U_i$
for some $V_j\in \cal V$ and we can 
take $L=V_j$. Therefore $F''/\rho^\diamond$ is isomorphic
to $A^\diamond$.
\qed\end{proof}

\begin{example}
The distributivity (\ref{e_nasobeni}) is not 
true for $w\in A^\diamond$ in general. Indeed, let $a,b\in A$ 
be different and let $L=\{aa,bb\}$. Then $\lambda\in L\circ 
(a\cdot (a\vee b)\wedge b\cdot(a\vee b))$ but
$L\circ ((a\wedge b)\cdot (a\vee b))= \emptyset$.
\end{example}

\section{Canonical Automata}
\label{s:canonical}

In each level, we consider the canonical finite automaton of a 
given regular language. 
To show examples of three types of automata, we
consider the language $L=a^+b^+$ over the alphabet $A=\{a,b\}$.

\medskip

\noindent {\bf Monoids.}
We considered the structure $(2^{A^*},A,\circ)$ defined 
by~(\ref{e-monoidy}). It is called here the {\it canonical 
 semiautomaton} on $A$.
 Given a regular language $L$ over $A$, 
 we can generate a subsemiautomaton by $L$ 
 in  $(2^{A^*},A,\circ)$ called  the 
 {\it canonical finite semiautomaton} 
 of $L$; namely
 $$\mathcal D_L = (\,\{\ L\circ u\, 
 \mid \, u\in A^*\, \}, A, \circ\ )\, .$$ 
 It is really finite due to
 Proposition~\ref{p:brzoz}.
 Notice that $L\circ u = u^{-1}L$ for all $u \in A^*$.
 Taking $L$ as the unique
 initial state and $T=\{\, L\circ u\, \mid\, \lambda\in L\circ u\, \}$ 
 as the set of all
 final states, we get the {\it canonical finite automaton} of $L$.

 \begin{proposition}[\cite{yu}]\label{p:brzoz}
 Given a regular language $L$ over the alphabet $A$, the automaton
 $\mathsf D_L=(\,\{u^{-1}L \mid u\in A^*\},A,\circ,L,T\,)$ is finite and accepts $L$.
 \end{proposition}

\begin{example}
In the canonical finite automaton $\mathsf{D}_L$ of the language
$L=a^+b^+$, we have four states 
$L=a^+b^+$, $K=a^{-1}L=a^*b^+$, $b^{-1}L=\emptyset$ and $b^{-1}K=b^*$.
There is just one state containing the empty word, namely 
the state $b^*$.  Thus $T=\{b^*\}$.
The automaton is depicted on Figure~\ref{f:automaton}.
\begin{figure}[ht]
\begin{center}
\begin{picture}(40,41)(0,-22)
    \unitlength=.5\unitlength
    \node[NLangle=0.0,Nmarks=i,iangle=180](n0)(0,0){$L$}
    \node[NLangle=0.0](n2)(25,25){$K$}
    \node[NLangle=0.0](n1)(25,-25){$\emptyset$}
    \node[NLangle=0.0](n6)(50,0){$b^*$}
\drawloop[loopangle=-30](n1){$a$}
\drawloop[loopangle=210](n1){$b$}
\drawloop[loopangle=30](n2){$a$}
\drawloop[loopangle=0](n6){$b$}
\drawedge(n0,n1){$b$}
\drawedge(n0,n2){$a$}
\drawedge(n2,n6){$b$}
\drawedge(n6,n1){$a$}
\rmark(n6)
\end{picture}
\caption{The canonical finite automaton of the language
$L=a^+b^+$.}
\label{f:automaton}
\end{center}
\end{figure}
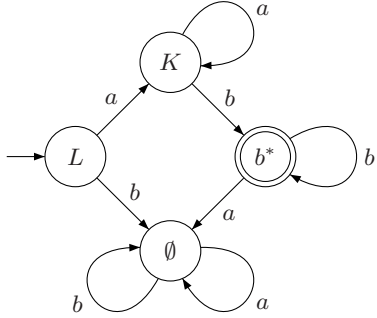
\vskip -0.8cm
\end{example}

\noindent {\bf Semirings.}
The structure $(2^{A^*},A,\circ,\cap)$ forms 
the {\it canonical meet semiautomaton}  on $A$. 
Moreover, given a regular language $L$ over $A$, we can generate 
 by $L$  in  $(2^{A^*},A,\circ,\cap)$  the 
 {\it canonical finite meet semiautomaton} 
 of $L$; namely
 $$\mathcal M_L=(\,\{\, L\circ U  \mid \, 
 U\in A^\square\ \}, A, \circ, \cap \, )\, .$$ 
 Taking $L$ as the unique
 initial state and  
 all states containing $\lambda$ as the set of all
 final states, we get the {\it canonical finite meet automaton} $\mathsf M_L$
 of $L$. 
 
 \begin{example}
To construct the canonical finite meet automaton $\mathsf{M}_L$ 
of the language $L=a^+b^+$ we need to consider
all possible intersections of states from 
$\mathcal{D}_L$.
There are two new states: the intersection $K\cap b^*=b^+$ 
and the intersection of the empty
system $\bigcap_{\emptyset}=A^*$.
The canonical finite meet automaton 
is depicted on Figure~\ref{f:meetautomaton}.
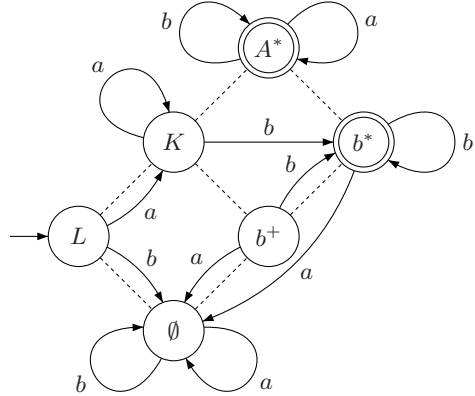
\begin{figure}[ht]
\begin{center}
  \begin{picture}(70,56)(-14,-22)
    \unitlength=.5\unitlength
\node[NLangle=0.0,Nmarks=i,iangle=180](n0)(0,0){$L$}
   \node[NLangle=0.0](n2)(25,25){$K$}
    \node[NLangle=0.0](n3)(50,0){$b^+$}
    \node[NLangle=0.0](n1)(25,-25){$\emptyset$}
    \node[NLangle=0.0](n5)(50,50){$A^*$}
    \node[NLangle=0.0](n6)(75,25){$b^*$}
\drawloop[loopangle=-30](n1){$a$}
\drawloop[loopangle=210](n1){$b$}
\drawloop[loopangle=15](n5){$a$}
\drawloop[loopangle=165](n5){$b$}
\drawloop[loopangle=135](n2){$a$}
\drawloop[loopangle=0](n6){$b$}
\drawedge[curvedepth=5](n0,n1){$b$}
\drawedge[curvedepth=-5,ELside=r](n0,n2){$a$}
\drawedge(n2,n6){$b$}
\drawedge[curvedepth=-5,ELside=r](n3,n1){$a$}
\drawedge[curvedepth=5](n3,n6){$b$}
\drawedge[curvedepth=10](n6,n1){$a$}
\drawedge[AHnb=0,dash={1 1}0](n0,n1){}
\drawedge[AHnb=0,dash={1 1}0](n0,n2){}
\drawedge[AHnb=0,dash={1 1}0](n1,n3){}
\drawedge[AHnb=0,dash={1 1}0](n2,n3){}
\drawedge[AHnb=0,dash={1 1}0](n2,n5){}
\drawedge[AHnb=0,dash={1 1}0](n3,n6){}
\drawedge[AHnb=0,dash={1 1}0](n5,n6){}
\rmark(n6)
\rmark(n5)
\end{picture}
\caption{The canonical finite meet automaton of the language
$L=a^+b^+$.}
\label{f:meetautomaton}
\end{center}
\vskip -0.6cm
\end{figure}

\noindent
Dashed lines indicate the inclusion relation on the set of all states. The inclusion relation
completely describes a semilattice structure of the meet automaton $\mathsf{M}_L$.
\end{example}

\noindent {\bf Lattice algebras}.
The structure $(2^{A^*},A,\circ,\cap,\cup)$ forms the 
{\it canonical lattice semiautomaton} 
 on $A$. Moreover, given a regular language $L$ over $A$, we can 
 generate 
 by $L$  in  $(2^{A^*},A,\circ,\cap,\cup)$  the 
 {\it canonical finite lattice semiautomaton} 
 of $L$; namely 
 $$\mathcal L_L=(\,\{\,L\circ {\mathcal U} \mid 
 {\mathcal U} \in A^\diamond\}, A,\circ,
 \cap,\cup\,).$$
 This structure is already mentioned in Kl\'ima~\cite{klima}.
 Taking $L$ as the unique
 initial state and 
 all states containing $\lambda$ as the set of all
 final states, we get the {\it canonical finite lattice 
 automaton} $\mathsf L_L$ of $L$.

\begin{example} 
We consider the canonical finite lattice automaton $\mathsf{L}_L$ 
of the language
$L=a^+b^+$, which is depicted on Figure~\ref{f:latticeautomaton}.
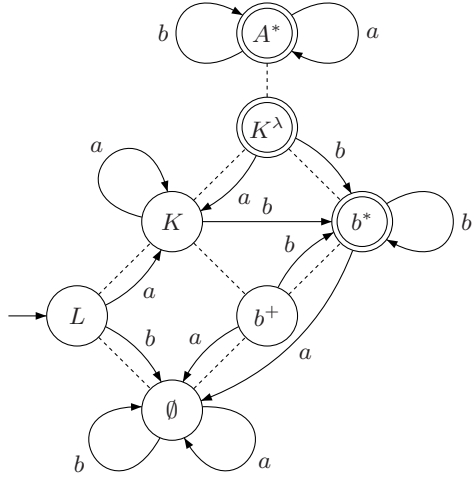
\begin{figure}[ht]
\begin{center}
  \begin{picture}(80,68)(-15,-24)
    \unitlength=.5\unitlength
\node[NLangle=0.0,Nmarks=i,iangle=180](n0)(0,0){$L$}
    \node[NLangle=0.0](n2)(25,25){$K$}
    \node[NLangle=0.0](n3)(50,0){$b^+$}
    \node[NLangle=0.0](n1)(25,-25){$\emptyset$}
\node[NLangle=0.0](n4)(50,50){$K^\lambda$}
\node[NLangle=0.0](n5)(50,75){$A^*$}
\node[NLangle=0.0](n6)(75,25){$b^*$}
\drawloop[loopangle=-30](n1){$a$}
\drawloop[loopangle=210](n1){$b$}
\drawloop[loopangle=0](n5){$a$}
\drawloop[loopangle=180](n5){$b$}
\drawloop[loopangle=135](n2){$a$}
\drawloop[loopangle=0](n6){$b$}
\drawedge[curvedepth=5](n0,n1){$b$}
\drawedge[curvedepth=-5,ELside=r](n0,n2){$a$}
\drawedge(n2,n6){$b$}
\drawedge[curvedepth=-5,ELside=r](n3,n1){$a$}
\drawedge[curvedepth=5](n4,n2){$a$}
\drawedge[curvedepth=5](n4,n6){$b$}
\drawedge[curvedepth=5](n3,n6){$b$}
\drawedge[curvedepth=10](n6,n1){$a$}
\drawedge[AHnb=0,dash={1 1}0](n0,n1){}
\drawedge[AHnb=0,dash={1 1}0](n0,n2){}
\drawedge[AHnb=0,dash={1 1}0](n1,n3){}
\drawedge[AHnb=0,dash={1 1}0](n2,n3){}
\drawedge[AHnb=0,dash={1 1}0](n2,n4){}
\drawedge[AHnb=0,dash={1 1}0](n3,n6){}
\drawedge[AHnb=0,dash={1 1}0](n4,n6){}
\drawedge[AHnb=0,dash={1 1}0](n4,n5){}
\rmark(n6)
\rmark(n5)
\rmark(n4)
\end{picture}
\caption{The canonical finite lattice automaton of the language
$L=a^+b^+$.}
\label{f:latticeautomaton}
\end{center}
\vskip -0.3cm
\end{figure}
There is only one new state, namely $K^\lambda=K\cup b^*=K\cup\{\lambda\}$
in addition to the canonical finite meet automaton $\mathsf{M}_L$.
Now, the inclusion relation describes a lattice structure of $\mathsf{L}_L$.
\end{example}

\section{Syntactic structures}
\label{s:syntactic}

The basic tool of the algebraic language theory is the concept
of the syntactic monoid of a regular language. It is a certain finite quotient
of the free monoid on the corresponding alphabet. We recall here
its definition and its construction.
Then we consider modifications for the remaining two levels.
\smallskip

\noindent {\bf Monoids.} Given a regular language $L$ over the alphabet $A$, 
we define the
 {\it syntactic congruence} 
$\sim^*_L$ of $L$ on $A^*$ as follows:
 for $u,v\in A^*$, put $u\sim^*_L v$  if and only if 
 $$ (\, \forall\, p,q\in A^*)\ (\ puq\in L 
 \Longleftrightarrow pvq\in L\ )\,.$$ 
 
 The following is a folklore result.

\begin{proposition}\label{syntactic monoid} 
The relation $\sim^*_L$ is a congruence relation on $A^*$.
Moreover,
for $u,v\in A^*$, we have that $u\sim^*_L v$ if and only if\ 
$$(\, \forall\ p\in A^*)
\ (p^{-1}L) \circ u = (p^{-1}L)\circ v\, .$$ Therefore, the 
structure $A^*/\!\sim^*_L$ 
is isomorphic to the transformation monoid of the canonical finite 
semiautomaton
of $L$.
\end{proposition}

We present here a proof since it is a suitable preparation
for similar results in the next levels.

\begin{proof}
Clearly, the relation $\ml$ is reflexive, symmetric and transitive.
Furthermore, for $u,v,w\in A^*$, if $u \ml v$ 
then $uw \ml vw$ and $wu \ml wv$.
Clearly, the fact $u\sim^*_L v$ is equivalent to
$(\, \forall\ p,q\in A^*)\, (\ q\in (pu)^{-1} L 
 \Longleftrightarrow q\in (pv)^{-1}L\ )$,
 which is
 $(\, \forall\ p\in A^*)\, (pu)^{-1} L =(pv)^{-1}L$,
 that is
 $(\, \forall\ p\in A^*)\, (p^{-1} L)\circ u = (p^{-1}L)\circ v$.
\qed\end{proof}

The structure $A^*/\!\sim^*_L$ is called
the {\it syntactic monoid} of $L$.
\smallskip

\noindent
{\bf Semirings.} 
Given a regular language $L$ over the alphabet $A$, 
we define the
{\it syntactic (semiring) congruence} 
$\sim^\square_L$ of $L$ on the semiring  $A^\square$ as follows:
 for $U=\{u_1,\dots,u_k\}, V=\{v_1,\dots,v_\ell\}
 \in  A^\square$, we put
 $U\sim^\square_L V$ if and only if  
 $$(\ \forall\ p,q\in A^*)\  
 (\ pu_1q\in L,\dots, pu_kq\in L \Longleftrightarrow pv_1q\in L,
 \dots,pv_\ell q\in L\  )\,.$$ 

\begin{proposition}[\cite{lp-syntactic}]
\label{syntactic semiring}
The relation $\sim^\square_L$ is a congruence relation on 
$A^\square$.
Moreover,
for $U,V\in A^\square$, we have that $U\sim^\square_L V$ if and only if
$$  
(\, \forall\ p\in A^*)\ (p^{-1}L) \circ U = (p^{-1}L)\circ V\,.$$
\end{proposition}

\begin{proof} 
To show that the relation $\sl$ is a congruence relation on
$A^\square$ is easy and similar to the case of monoids.
Clearly, the fact $U\sl V$ is equivalent to
$$(\, \forall\, p,q\in A^*) \ q\in (pu_1)^{-1} L \cap \dots \cap
(pu_k)^{-1} L
 \Longleftrightarrow q\in (pv_1)^{-1}L  \cap \dots \cap
 (pv_\ell)^{-1}L\,.$$
 The last formula can be written as
 $$(\,\forall\ p\in A^*)\, (pu_1)^{-1} L \cap \dots \cap (pu_k)^{-1} L
 =(pv_1)^{-1} L \cap \dots \cap (pv_\ell)^{-1} L\,,$$
 which is 
 $(\,\forall\ p\in A^*) \ p^{-1} L\circ U =p^{-1}L\circ V$.
\qed\end{proof}

Note that  one can show (see~\cite{lp-examples})
that the structure $A^\square/\!\sim_L^\square$ is isomorphic
to  the transformation semiring of the whole canonical finite meet 
semiautomata $\mathsf M_L$ of $L$. 
The structure $A^\square/\!\sim_L^\square$
is called the {\it syntactic semiring} of $L$.
\medskip

Numerous examples of syntactic semirings can be found e.g. in~\cite{lp-examples}.
In~\cite{lp-syntactic}
it is described how one can compute the 
syntactic semiring algorithmically from the 
syntactic monoid. 
For the handmade computations we can use Proposition~\ref{syntactic semiring}.

\begin{example} 
Consider again the language
$L=a^+b^+$. We can choose the words $\lambda$, $a$, $b$, $ab$, and $ba$ 
to represent five different transformations. There are no others, because 
both $a$ and $b$ are idempotent elements of both syntactic monoid and syntactic 
semiring and $ba$ is a zero element. Moreover, $ba$ is
the smallest element in the syntactic semiring, because $ba$ transforms
all states, with exception of $A^*$, to the state $\emptyset$. So,
if we want to compute all elements of the syntactic
semiring, it is enough to consider only intersections of the elements
$\lambda$, $a$, $b$ and $ab$. The crucial observation is that both $\lambda\wedge ab$ and 
$a\wedge b$ give the same transformation 
as well as the intersection of any triple of elements.
Hence in the syntactic semiring there are, besides the element $\top$ and elements
$\lambda$, $a$, $b$, $ab$, and $ba$, just five elements given by intersections
$\lambda \wedge a$, $\lambda \wedge b$, $\lambda\wedge ab$, $ a\wedge ab$ and $b\wedge ab$.
In Table~\ref{t:syntactic-semiring} 
\begin{table}[ht]
\begin{center}
\begin{tabular}{c|ccc|ccc}
&$L$&$K$&$b^*$&$b^+$&$A^*$&$\emptyset$\\
\hline
$\lambda$&$L$&$K$&$b^*$&$b^+$&$A^*$&$\emptyset$\\
$a$&$K$&$K$&$\emptyset$&$\emptyset$&$A^*$&$\emptyset$\\
$b$&$\emptyset$&$b^*$&$b^*$&$b^*$&$A^*$&$\emptyset$\\
$ab$&$b^*$&$b^*$&$\emptyset$&$\emptyset$&$A^*$&$\emptyset$\\
$ba$&$\emptyset$&$\emptyset$&$\emptyset$&$\emptyset$&$A^*$&$\emptyset$\\
\hline
$\top$ &$A^*$&$A^*$&$A^*$&$A^*$&$A^*$&$A^*$\\
$\lambda\wedge a$&$L$&$K$&$\emptyset$&$\emptyset$&$A^*$&$\emptyset$\\
$\lambda\wedge b$&$\emptyset$&$b^+$&$b^*$&$b^+$&$A^*$&$\emptyset$\\
$\lambda \wedge ab$&$\emptyset$&$b^+$&$\emptyset$&$\emptyset$&$A^*$&$\emptyset$\\
$a\wedge ab$&$b^+$&$b^+$&$\emptyset$&$\emptyset$&$A^*$&$\emptyset$\\
$b\wedge ab$&$\emptyset$&$b^*$&$\emptyset$&$\emptyset$&$A^*$&$\emptyset$\\
\hline
\end{tabular}
\end{center}
\caption{The transformations of $\mathcal{M}_L$ for the language $L=a^+b^+$.}
\label{t:syntactic-semiring}
\end{table}
we present how all these elements transform
the canonical finite meet automaton. The semilattice part of the syntactic semiring
is fully described by Figure~\ref{f:syntactic-semiring}.
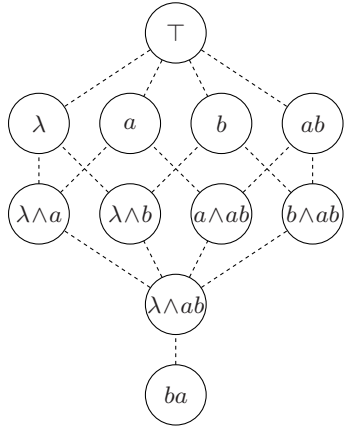
\begin{figure}[ht]
\begin{center}
  \begin{picture}(72,60)(-14,-5)
    \unitlength=.5\unitlength
\node[NLangle=0.0](p1)(0,48){{\footnotesize $\lambda\!\wedge\! a$}}
   \node[NLangle=0.0](p2)(24,48){{\footnotesize $\lambda\!\wedge\! b$}}
    \node[NLangle=0.0](p3)(48,48){{\footnotesize $a\!\wedge\! ab$}}
    \node[NLangle=0.0](p4)(72,48){{\footnotesize $b\! \wedge\! ab$}}
    \node[NLangle=0.0](d1)(0,72){{\footnotesize $\lambda$}}
   \node[NLangle=0.0](d2)(24,72){{\footnotesize $a$}}
    \node[NLangle=0.0](d3)(48,72){{\footnotesize $b$}}
    \node[NLangle=0.0](d4)(72,72){{\footnotesize $ab$}}
    \node[NLangle=0.0](v)(36,96){{\footnotesize $\top$}}
    \node[NLangle=0.0](s2)(36,24){{\footnotesize $\lambda\!\wedge\!  ab$}}
    \node[NLangle=0.0](s1)(36,0){{\footnotesize $ba$}}
\drawedge[AHnb=0,dash={1 1}0](s1,s2){}
\drawedge[AHnb=0,dash={1 1}0](s2,p1){}
\drawedge[AHnb=0,dash={1 1}0](s2,p2){}
\drawedge[AHnb=0,dash={1 1}0](s2,p3){}
\drawedge[AHnb=0,dash={1 1}0](s2,p4){}
\drawedge[AHnb=0,dash={1 1}0](v,d1){}
\drawedge[AHnb=0,dash={1 1}0](v,d2){}
\drawedge[AHnb=0,dash={1 1}0](v,d3){}
\drawedge[AHnb=0,dash={1 1}0](v,d4){}
\drawedge[AHnb=0,dash={1 1}0](p1,d1){}
\drawedge[AHnb=0,dash={1 1}0](p1,d2){}
\drawedge[AHnb=0,dash={1 1}0](p4,d3){}
\drawedge[AHnb=0,dash={1 1}0](p4,d4){}
\drawedge[AHnb=0,dash={1 1}0](p2,d1){}
\drawedge[AHnb=0,dash={1 1}0](p2,d3){}
\drawedge[AHnb=0,dash={1 1}0](p3,d2){}
\drawedge[AHnb=0,dash={1 1}0](p3,d4){}
\end{picture}
\caption{The semilattice order of the syntactic semiring of the language $L=a^+b^+$.}
\label{f:syntactic-semiring}
\end{center}
\end{figure}
Notice that for the computation of the syntactic semiring we do not 
need to know all the information from Table~\ref{t:syntactic-semiring}.
For example, if a term $U\in F'$ acts on the state $b^+$, then the image is
the intersection of images of the states $K$ and $b^*$.
Moreover, the images of states $\emptyset$ and $A^*$ are clear. 
Thus we need to work only with first three columns. 
\end{example}

\noindent {\bf Lattice algebras}. 
Given a regular language $L$ over the alphabet $A$, 
we define the
so-called {\it syntactic (lattice) congruence} 
$\sim^\diamond_L$ of $L$ on $A^\diamond$ as follows:
for 
${\mathcal U}=\{U_1,\dots,U_k\},{\mathcal V}=\{V_1,\dots,V_\ell\}
\in A^\diamond$ we put
$\mathcal U\sim_L^\diamond \mathcal V$ if and only if,  
for every $ p,q\in A^*$, the condition
$$
  pU_1q\subseteq L \text{ or } \dots \text{ or } 
 pU_kq\subseteq L\ $$
 is equivalent to 
$$ pV_1q \subseteq L \text{ or } \dots 
 \text{ or } pV_\ell q\subseteq L \,.$$

\begin{proposition}\label{p:syntactic lattice algebra}
The relation $\sim^\diamond_L$ is a congruence relation on 
$A^\diamond$.
Moreover,
for ${\mathcal U},{\mathcal V}\in A^\diamond$, 
it holds that ${\mathcal U}\sim^\diamond_L {\mathcal V}$ 
if and only if $$(\,\forall\ p\in A^*)\ p^{-1}L \circ \mathcal U 
= p^{-1}L\circ \mathcal V\, .$$
\end{proposition}

\begin{proof} To show that the relation $\ll$ is a congruence relation on
$A^\diamond$ is easy and similar to the case of monoids.

Let ${\mathcal U}, {\mathcal V}\in A^\diamond$ are of the form
$$\mathcal U=\{U_1,\dots,U_k\},\ \text{ where }\
 U_1=\{u_{1,1},\dots,u_{1,r_1}\},\dots,
 U_k=\{u_{k,1},\dots,u_{k,r_k}\}\,,$$
 $$\mathcal V=\{V_1,\dots,V_\ell\},\  \text{ where }\
  V_1=\{v_{1,1},\dots,v_{1,s_1}\},\dots,
 V_\ell=\{v_{\ell,1},\dots,v_{\ell,s_\ell}\}\, .$$ 
Clearly, ${\mathcal U}\ll {\mathcal V}$ is equivalent to
$(\, \forall\ p,q\in A^*)$ $$ q\in 
((pu_{1,1})^{-1} L \cap \dots \cap (pu_{1,r_1})^{-1} L)
\cup\dots\cup
((pu_{k,1})^{-1} L \cap \dots \cap (pu_{k,r_k})^{-1} L)
$$$$\Longleftrightarrow q\in 
(pv_{1,1})^{-1} L \cap \dots \cap (pv_{1,s_1})^{-1} L
\cup\dots\cup
(pv_{\ell,1})^{-1} L \cap \dots \cap (pv_{\ell,s_\ell})^{-1} L\,,$$
 which is
 $$(\,\forall\ p\in A^*)\ 
 ((pu_{1,1})^{-1} L \cap \dots \cap (pu_{1,r_1})^{-1} L)
\cup\dots\cup
((pu_{k,1})^{-1} L \cap \dots \cap (pu_{k,r_k})^{-1} L)
$$$$=((pv_{1,1})^{-1} L \cap \dots \cap (pv_{1,s_1})^{-1} L)
\cup\dots\cup
((pv_{\ell,1})^{-1} L \cap \dots \cap (pv_{\ell,s_k})^{-1} L)
\, ,$$
 that is
 $(\,\forall\ p\in A^*) \ p^{-1} L\circ \mathcal U 
 =p^{-1}L\circ \mathcal V$.
\qed\end{proof}

The structure $A^\diamond/\ll$ is called the 
{\it syntactic lattice algebra} 
 of $L$.
\medskip

Note that in this third level it is not true that 
the structure $A^\diamond/\!\sim_L^\diamond$ 
is isomorphic to  the transformation lattice algebra
of the whole canonical lattice semiautomaton
$\mathcal L_L$ of $L$ as mentioned in the next example.

\begin{example}
Now we present the syntactic lattice algebra
of the language $L=a^+b^+$. 

First of all, we could mentioned an interesting fact:
the terms $\lambda \wedge ab$ and 
$a\wedge b$ transform $\mathcal L_L$ in a different way, namely 
$K^\lambda \circ (\lambda \wedge ab)=b^*$
and $K^\lambda \circ (a \wedge b)=b^+$. 
However these two terms
$\lambda \wedge ab$ and $a\wedge b$ give the same element 
in  the syntactic semiring of $L$, because 
they transform the states from $\mathcal D_L$ in the same way.
In other words, $\lambda \wedge ab\, \rho_L^\diamond\, a \wedge b$. 
This example just recalls the observation from Proposition~\ref{p:syntactic lattice algebra}, 
that we need to check the images of the three states $L$, $K$ and $b^*$ only.

We can start from the syntactic semiring of $L$, 
since the syntactic lattice algebra can be viewed as an extension of the syntactic semiring
by adding joins.
Thus we need to compute joins of all elements described in Table~\ref{t:syntactic-semiring}.
This can be done by a brute force algorithm, which gives 
Table~\ref{t:syntactic-lattice}. 
\begin{table}[ht]
\begin{center}
\begin{tabular}{cp{.1cm}|cccp{1cm}cp{.1cm}|ccc}
&&$L$&$K$&$b^*$& & &&$L$&$K$&$b^*$ \\
\cline{1-5}\cline{7-11}
$\lambda$&&$L$&$K$&$b^*$&  &$(\lambda\wedge a)\vee (b\wedge ab)$&&$L$&$K^\lambda$&$\emptyset$\\
$a$&&$K$&$K$&$\emptyset$&   &$(\lambda\wedge a)\vee b$ &&$L$&$K^\lambda$&$b^*$  \\
$b$&&$\emptyset$&$b^*$&$b^*$&  & $(\lambda\wedge a)\vee ab$ &&$K^\lambda$&$K^\lambda$&$\emptyset$  \\
$ab$&&$b^*$&$b^*$&$\emptyset$&   & $(\lambda\wedge b)\vee(a\wedge ab) $ &&$b^+$&$b^+$&$b^*$   \\
$ba$&&$\emptyset$&$\emptyset$&$\emptyset$&  & $(\lambda\wedge b)\vee a  $ &&$K$&$K^\lambda$&$b^*$   \\
$\top$ &&$A^*$&$A^*$&$A^*$&       & $(\lambda\wedge b)\vee ab $ &&$b^*$&$b^*$&$b^*$    \\
$\lambda\wedge a$&&$L$&$K$&$\emptyset$&    & $(a\wedge ab)\vee (b\wedge ab) $ &&$b^+$&$b^*$&  $\emptyset$   \\
$\lambda\wedge b$&&$\emptyset$&$b^+$&$b^*$&   &  $(a\wedge ab)\vee \lambda $ &&$K$&$K$&   $b^*$    \\
$\lambda \wedge ab$&&$\emptyset$&$b^+$&$\emptyset$&   & $(a\wedge ab)\vee b $ &&$b^+$&$b^*$&$b^*$   \\
$a\wedge ab$&&$b^+$&$b^+$&$\emptyset$&    & $(b\wedge ab)\vee a $ &&$K$&$K^\lambda$&   $\emptyset$   \\
$b\wedge ab$&&$\emptyset$&$b^*$&$\emptyset$&  & $\lambda \vee ab$  &&$K^\lambda$&$K^\lambda$&   $b^*$   \\
\cline{1-5}\cline{7-11}
\end{tabular}
\end{center}
\caption{The transformations of $\mathcal{L}_L$ for the language $L=a^+b^+$.}
\label{t:syntactic-lattice}
\end{table}

To see that the computation is complete, we have to add some basic observations.
At first, one can check the following equalities $\lambda=(\lambda\wedge a)\vee (\lambda \wedge b)$,
$a=(\lambda\wedge a)\vee (a \wedge ab)$ and $b=(\lambda\wedge b)\vee (b \wedge ab)$. 
Therefore we can remove elements $\lambda$, $a$ and $b$ from the generating set.
Since the elements $\top$, $\lambda \wedge ab$ and $ba$ are comparable with the others,
we do not obtain new elements adding these element into the joins.
Thus, we need to compute the joins for five elements $\lambda\wedge a$,  $\lambda\wedge b$, $a\wedge ab$
$b\wedge ab$ and $ab$. 

We observe that $a \wedge ab$ transforms the state $L$ to $b^+$, and that 
no other image of $L$ under applications $\lambda\wedge a$, $\lambda\wedge b$, $b\wedge ab$ contains 
the word $b$. This means that $a \wedge ab$ can not be covered by a join of elements 
$\lambda\wedge a$, $\lambda\wedge b$, $b\wedge ab$.
In the similar way,  $b\wedge ab$ transforms $K$ to $b^*$ which contains $\lambda$, and therefore 
$b\wedge ab$ can not be covered by a join of elements 
$\lambda\wedge a$, $\lambda\wedge b$, $a\wedge ab$.
To see that both $\lambda\wedge a$ and $\lambda\wedge b$ can not be covered by a 
join of the  others elements from the following ones
$\lambda\wedge a$,  $\lambda\wedge b$, $a\wedge ab$
$b\wedge ab$, $ab$, we just mention that $K\circ (\lambda\wedge a)=K$ contains $ab$ and
$b^*\circ (\lambda\wedge b)=b^*$ contains $\lambda$. 

From the observations from the previous paragraph
we can state that joins of elements 
$\lambda\wedge a,  \lambda\wedge b, a\wedge ab, b\wedge ab$ are pairwise different
elements of the syntactic lattice algebra of $L$. So we obtain 15 elements in this way. 
If we add the element $ab$ into some of these joins, then we can remove from this join both
$a\wedge ab$ and $ b\wedge ab$ if they occur. So, we obtain additionally 
four elements $ab$, $ab\vee (\lambda \wedge a)$,  $ab\vee (\lambda \wedge b)$,
$ab\vee (\lambda \wedge a)\vee (\lambda \wedge b)=ab\vee \lambda$.

Altogether, the syntactic lattice of $L$ consists of 22 elements (see Figure~\ref{f:lattice-algebra}):
$\perp=ba$, $\lambda \wedge ab$, 15 elements described above as joins of elements 
$\lambda\wedge a,  \lambda\wedge b, a\wedge ab, b\wedge ab$, and finally the elements 
$ab$, $ab\vee (\lambda \wedge a)$,  $ab\vee (\lambda \wedge b)$, $ab\vee \lambda$, $\top$.
\begin{figure}[ht]
\begin{center}
  \begin{picture}(120,93)(-14,-5)
    \unitlength=.5\unitlength
\node[NLangle=0.0](p1)(0,48){{\footnotesize $\lambda\!\wedge\! a$}}
   \node[NLangle=0.0](p2)(24,48){{\footnotesize$\lambda\!\wedge\! b$}}
    \node[NLangle=0.0](p3)(48,48){{\footnotesize$a\!\wedge\! ab$}}
    \node[NLangle=0.0](d1)(0,72){{\footnotesize$\lambda$}}
   \node[NLangle=0.0](d2)(24,72){{\footnotesize$a$}}
    \node[NLangle=0.0](d3)(48,72){}
    \node[NLangle=0.0](v)(24,96){}
    \node[NLangle=0.0](s2)(24,24){{\footnotesize$\lambda\!\wedge\!  ab$}}
    \node[NLangle=0.0](s1)(24,0){{\footnotesize$\perp$}}
    \node[NLangle=0.0](pp1)(72,72){}
   \node[NLangle=0.0](pp2)(96,72){{\footnotesize$b$}}
    \node[NLangle=0.0](pp3)(120,72){}
    \node[NLangle=0.0](ps)(96,48){{\footnotesize$b\! \wedge\! ab$}}
    \node[NLangle=0.0](pd1)(72,96){}
   \node[NLangle=0.0](pd2)(96,96){}
    \node[NLangle=0.0](pd3)(120,96){}
    \node[NLangle=0.0](pv)(96,120){}
    \node[NLangle=0.0](dp3)(168,96){{\footnotesize$ab$}}
   \node[NLangle=0.0](dd2)(144,120){}
    \node[NLangle=0.0](dd3)(168,120){}
    \node[NLangle=0.0](dv)(144,144){}
  \node[NLangle=0.0](vvv)(144,168){{\footnotesize$\top$}}
\drawedge[AHnb=0,dash={1 1}0](s1,s2){}
\drawedge[AHnb=0,dash={1 1}0](s2,p1){}
\drawedge[AHnb=0,dash={1 1}0](s2,p2){}
\drawedge[AHnb=0,dash={1 1}0](s2,p3){}
\drawedge[AHnb=0,dash={1 1}0](v,d1){}
\drawedge[AHnb=0,dash={1 1}0](v,d2){}
\drawedge[AHnb=0,dash={1 1}0](v,d3){}
\drawedge[AHnb=0,dash={1 1}0](p1,d1){}
\drawedge[AHnb=0,dash={1 1}0](p1,d2){}
\drawedge[AHnb=0,dash={1 1}0](p3,d3){}
\drawedge[AHnb=0,dash={1 1}0](p2,d1){}
\drawedge[AHnb=0,dash={1 1}0](p2,d3){}
\drawedge[AHnb=0,dash={1 1}0](p3,d2){}
\drawedge[AHnb=0,dash={1 1}0](ps,pp1){}
\drawedge[AHnb=0,dash={1 1}0](ps,pp2){}
\drawedge[AHnb=0,dash={1 1}0](ps,pp3){}
\drawedge[AHnb=0,dash={1 1}0](pv,pd1){}
\drawedge[AHnb=0,dash={1 1}0](pv,pd2){}
\drawedge[AHnb=0,dash={1 1}0](pv,pd3){}
\drawedge[AHnb=0,dash={1 1}0](pp1,pd1){}
\drawedge[AHnb=0,dash={1 1}0](pp1,pd2){}
\drawedge[AHnb=0,dash={1 1}0](pp3,pd3){}
\drawedge[AHnb=0,dash={1 1}0](pp2,pd1){}
\drawedge[AHnb=0,dash={1 1}0](pp2,pd3){}
\drawedge[AHnb=0,dash={1 1}0](pp3,pd2){}
\drawedge[AHnb=0,dash={1 1}0](s2,ps){}
\drawedge[AHnb=0,dash={1 1}0](p1,pp1){}
\drawedge[AHnb=0,dash={1 1}0](p2,pp2){}
\drawedge[AHnb=0,dash={1 1}0](p3,pp3){}
\drawedge[AHnb=0,dash={1 1}0](d1,pd1){}
\drawedge[AHnb=0,dash={1 1}0](d2,pd2){}
\drawedge[AHnb=0,dash={1 1}0](d3,pd3){}
\drawedge[AHnb=0,dash={1 1}0](v,pv){}
\drawedge[AHnb=0,dash={1 1}0](dp3,dd2){}
\drawedge[AHnb=0,dash={1 1}0](dp3,dd3){}
\drawedge[AHnb=0,dash={1 1}0](dd3,dv){}
\drawedge[AHnb=0,dash={1 1}0](dd2,dv){}
\drawedge[AHnb=0,dash={1 1}0](dv,vvv){}
\drawedge[AHnb=0,dash={1 1}0](pp3,dp3){}
\drawedge[AHnb=0,dash={1 1}0](pd2,dd2){}
\drawedge[AHnb=0,dash={1 1}0](pd3,dd3){}
\drawedge[AHnb=0,dash={1 1}0](pv,dv){}
\end{picture}
\caption{The order of the syntactic lattice algebra of the language $L=a^+b^+$.}
\label{f:lattice-algebra}
\end{center}
\end{figure}
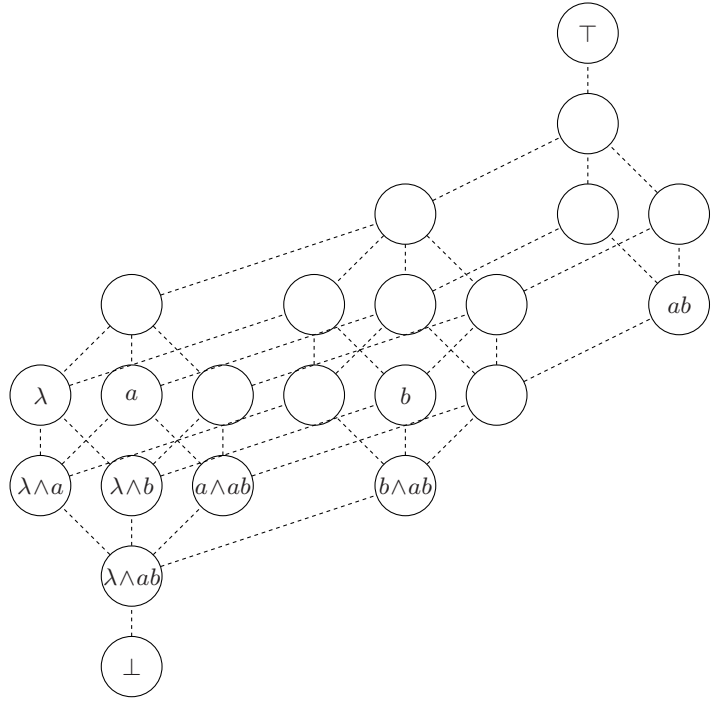
\end{example}

\section{General algebras}
 
 The Eilenberg like theorems establish 
 bijections between certain varieties of regular languages and
 pseudovarieties of certain algebraic systems.
Not every finite monoid is isomorphic to a syntactic one,  we have to 
generate the appropriate pseudovariety. Similarly in remaining
levels.
\medskip

\noindent {\bf Monoids}. 
Here one considers varieties of languages
and pseudovarieties of finite monoids. The Eilenberg theorem can be find 
in e.g.~\cite{pi}.
\medskip

\noindent {\bf Semirings}.
Here one considers the so-called conjunctive varieties of languages
and pseudovarieties of finite semirings. For more details see e.g.~\cite{lp-examples}.
\medskip

\noindent {\bf Lattice algebras}.
The following new definition of a notion of lattice algebras 
is a part of an effort of formulation of Eilenberg like theorem using the notion of
syntactic lattice algebra. Such a theorem is not formulated or even proved in this paper. 
Nevertheless, we try here to characterize the finite factors of $A^\diamond$.

A {\it lattice algebra} is 8-tuple $(K,\wedge,\vee,\cdot, P,\bot,\top,1)$ where 
$(K,\wedge,\vee)$ is a bounded distributive lattice with the bottom element $\bot$ 
and the top element $\top$,
$(K,\cdot,1)$ is a monoid with right zero elements $\bot$ and $\top$, 
$P$ is a finite subset of $K$ such that the 
lattice $(K,\wedge,\vee)$ is generated by the set of all 
products of elements from $P$
and $\top \cdot p=\top$ and $\bot\cdot p=\bot$ hold for $p\in P$,
and finally such that the distributivity laws
$$q\cdot (r\wedge s)=q\cdot r \wedge q\cdot s,\ \
q\cdot (r\vee s)=q\cdot r \vee q\cdot s\, ,$$
$$(q\wedge r)\cdot p= q\cdot p\wedge r \cdot p,\ \
(q\vee r)\cdot p=q\cdot p\vee r \cdot p$$
hold for all $q,r,s\in K$ and $p\in P$.

Notice that, considering
$A^\diamond$, take $P$ equal to the image of $A$, $\top$
the image of $\{\emptyset\}$ and $\bot$
the image of $\emptyset$. 

\section{Characterizing Reversible Languages}

We consider the class of all reversible languages 
(see Golovkins and Pin~\cite{golovkins-pin}).
We present them using the Ambainis and Freivalds condition
(see~\cite{ambainis-freivalds}). 

\begin{proposition}[{\cite{ambainis-freivalds,golovkins-pin}}]
Let $L$ be a regular language over an alphabet $A$. 
Then $L$ is recognized by a reversible  automaton if and only if the following condition for the canonical 
automaton of $L$ holds:
\begin{equation}
(\, \forall\, x,y\in A^*, f, g \in Q\, )\ f\not=g, 
f\circ x =g = g\circ x \implies g\circ y =g \, .
\label{e-AF}
\end{equation} 
\end{proposition}

Note that a condition from the previous statement is 
usually formulated in a different way, namely
that the canonical automaton of $L$ does not contain the following configuration, with
$f\not= g\not =h$.

\begin{figure}[ht]
\begin{center}
\begin{picture}(50,20)(0,-10)
    \unitlength=.5\unitlength
    \node(n1)(0,-10){$f$}
    \node(n2)(50,-10){$g$}
    \node(n3)(100,-10){$h$}
    \drawloop(n2){$x$}
    \drawedge(n1,n2){$x$}
    \drawedge(n2,n3){$y$}
\end{picture}
\caption{The forbidden configuration for reversible language.}
\label{f:reversible}
\end{center}
\end{figure}
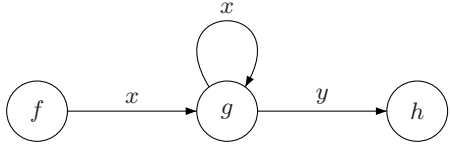

In~\cite{golovkins-pin}
the Ambainis-Freivalds condition (\ref{e-AF}) for the language $L$
was translated to a certain algebraic condition 
concerning the syntactic monoid of $L$
together with the image of $L$ in the syntactic homomorphism.
They also mention that this class is not closed under binary intersections nor 
unions. Therefore it is not an instance of any known Eilenberg correspondence.

Here we show an equivalent condition which is, in some sense, an identity 
for the canonical 
lattice algebra of the considered language.
We need the following classical notion from the semigroup theory. 
Each element $s$ in a finite semigroup has a unique idempotent element among its powers, 
which is denoted by $s^\omega$. So we use this notation for lattice algebra, where 
this operation ${(\underline{\phantom{s}}})^\omega$ is related to the operation of multiplication.
Moreover, in a fixed finite semigroup $S$, one can find natural number $m$ such that 
$s^\omega =s^m$ for every element $s\in S$. 

We are not going to define here the notion of an identity for 
(finite) lattice algebras in a full generality. 
Nevertheless, we use the concrete condition
\begin{equation}
x^\omega y \vee (x^\omega z \wedge t) =  
x^\omega y \vee (x^\omega t \wedge z)\ . 
\label{e-identity}
\end{equation}
It is {\it valid} in the syntactic lattice algebra $\mathcal L_L$ of the language $L$
if we get the same element of (\ref{e-identity}) on left and right sides after 
substituting $p\sim^\diamond_L$, $u\sim^\diamond_L$, 
$v\sim^\diamond_L$ and $w\sim^\diamond_L$ ($p,u,v,w\in A^*$),
for $x,y,z$ and $t$, respectively.

\begin{proposition}
The canonical automaton of a regular language $L$ satisfies (\ref{e-AF}) if and only if
the canonical lattice algebra of $L$ satisfies 
condition
(\ref{e-identity}).
\end{proposition}

\begin{proof}
To simplify notation we write simply $u$ instead of $u\sim^\diamond_L$,
for any $u\in A^*$.  
This simplification does not lead to a confusion, because
for a state of $K$ of the canonical semiautomaton  of a language $L$, 
by $K\circ (U\sim^\diamond_L)$ is meant $K\circ U$. 

Let $L$ be a regular language with the canonical automaton satisfying the 
condition (\ref{e-AF}).
Let  $p,u,v,w\in A^*$ be arbitrary words and
denote ${\mathcal U}=p^\omega u \vee (p^\omega v \wedge w),{\mathcal V}=p^\omega u \vee (p^\omega w \wedge v)$,
both from
$A^\diamond/{\sim^\diamond_L}$. 
Furthermore, let $s\in A^*$ be an arbitrary word and consider the state $K=s^{-1}L$ in 
the canonical automaton of $L$.
We need to show that $K \circ {\mathcal U}=K\circ {\mathcal V}$.
At first, assume that $K\circ p^\omega =K$. Then  
$K \circ {\mathcal U}=K\circ u\cup  (K\circ v \cap K\circ w)$ which is equal to $K\circ {\mathcal V}$.
Assume now that $K\circ p^\omega \not =K$, particularly $K \circ p \not = K$.
From the definition of $p^\omega$ we know that $(K\circ p^\omega)\circ p^\omega = K\circ p^\omega$. 
Since the canonical semiautomaton  $\mathcal D_L$ satisfies (\ref{e-AF}), 
we get that $(K\circ p^\omega )\circ y =K\circ p^\omega$
for every $y\in A^*$. Therefore, 
$K \circ {\mathcal U}=K\circ p^\omega \cup  (K\circ p^\omega \cap K\circ w)=K\circ p^\omega$
and  similarly we obtain 
$K \circ {\mathcal V}=K\circ p^\omega \cup  (K\circ p^\omega \cap K\circ v)=K\circ p^\omega$.

To prove the opposite implication, we consider a regular language
$L$ which has the forbidden configuration in 
its canonical semiautomaton $\mathcal D_L$ and then we show that  its 
canonical lattice algebra does not 
satisfy~(\ref{e-identity}).
Let  $f,g,h$ be states in $\mathcal D_L$ and $x,y\in A^*$ be words such that
$f\not =g \not= h$, $f\circ x =g = g\circ x$ and $h=g\circ y$. 
Recall that $f,g,h \subseteq A^*$, 
because there are left quotient of $L$.
Since $f\not =g$, there is a word $s\in A^*$ such that $s\in f$, $s\not \in g$ or 
$s\not \in f$, $s\in g$. Note that  the condition  $s\in f$ is equivalent to 
$\lambda \in s^{-1} f$, i.e. 
$\lambda\in f\circ s$.
Similarly, since $g\not = h$, there is a word $r$ such that $r\in g$, $r\not \in h$ or 
$r\not \in g$, $r\in h$.
Thus there are four cases to be discussed. In all these cases, the word 
$p=x$ is already fixed by the forbidden configuration.
For this $p$, we have $f \circ p^\omega =g$.

Case I) If $s\in f$, $s\not \in g$,  $r\in g$, $r\not \in h$ then  we put $u=s$, $v=r$ and $w=s$
and consequently we denote ${\mathcal U}= p^\omega u \vee (p^\omega v \wedge w)$, ${\mathcal V} = 
p^\omega u \vee (p^\omega w \wedge v)$.
Now we see that 
\begin{eqnarray*}
\lambda &\not \in& g \circ s= (f \circ p^\omega) \circ s=f \circ (p^\omega u)=f \circ (p^\omega w)\, , \\ 
\lambda & \in & g\circ r= (f \circ p^\omega)\circ  r=f \circ (p^\omega v)\, \text{ and }\\
\lambda & \in & f \circ s = f \circ w\, . 
\end{eqnarray*}
Therefore,
$\lambda \in f \circ p^\omega u \cup (f  \circ p^\omega v \cap f\circ w)=f\circ {\mathcal U}$
and $\lambda \not\in f \circ p^\omega u \cup (f  \circ p^\omega w \cap f\circ v)=f\circ {\mathcal V}$.
Hence $f\circ {\mathcal U}\not =f\circ {\mathcal V}$ and ${\mathcal U}$, ${\mathcal V}$ 
are different elements in the canonical lattice algebra of $L$.

Case II) If $s\in f$, $s\not \in g$,  $r\not\in g$, $r\in h$ then  we put $u=s$, $v=yr$ and $w=s$
and again ${\mathcal U}= p^\omega u \vee (p^\omega v \wedge w)$, ${\mathcal V} = 
p^\omega u \vee (p^\omega w \wedge v)$.
Now we have 
\begin{eqnarray*}
\lambda & \not \in & g \circ s= (f \circ p^\omega) \circ s=f \circ (p^\omega u)=f \circ (p^\omega w)\, ,\\  
\lambda & \in & h\circ r= (f \circ p^\omega y)\circ  r=f \circ (p^\omega v)\, \text{ and }\\
\lambda & \in & f \circ s = f \circ w\, .
\end{eqnarray*}
Therefore, $\lambda \in f \circ p^\omega u \cup (f  \circ p^\omega v \cap f\circ w)=f\circ {\mathcal U}$
and $\lambda \not\in f \circ p^\omega u \cup (f  \circ p^\omega w \cap f\circ v)=f\circ {\mathcal V}$.
This means that ${\mathcal U}\not ={\mathcal V}$ in the canonical lattice algebra of $L$.

Case III) If $s\not\in f$, $s\in g$,  $r\in g$, $r\not \in h$ then  we put $u=y r$, $v=s$, 
$w=pr$, ${\mathcal U}= p^\omega u \vee (p^\omega v \wedge w)$ and ${\mathcal V} = p^\omega u \vee (p^\omega w \wedge v)$.
Now we have 
\begin{eqnarray*}
\lambda & \not \in & h \circ r=(g\circ y)\circ r= g\circ u= 
(f \circ p^\omega) \circ u=f \circ (p^\omega u)\, , \\
\lambda & \in & g \circ s= g\circ v=(f \circ p^\omega) \circ v=
f \circ (p^\omega v)\, , \\   
\lambda & \in&  g\circ r= (f\circ p)\circ r= f\circ w
\text{ and } \\
\lambda & \not\in & f \circ s = f\circ v\, .
\end{eqnarray*}
Hence, $\lambda \in f \circ p^\omega u \cup (f  \circ p^\omega v \cap f\circ w)=f\circ {\mathcal U}$
and $\lambda \not\in f \circ p^\omega u \cup (f  \circ p^\omega w \cap f\circ v)=f\circ {\mathcal V}$.

Case IV) If $s\not\in f$, $s\in g$,  $r\not\in g$, $r\in h$ then  we put $u=r$, $v=s$ 
and $w=ps$ and consequently ${\mathcal U}= p^\omega u \vee (p^\omega v \wedge w)$, ${\mathcal V} = 
p^\omega u \vee (p^\omega w \wedge v)$.
Now we see that 
\begin{eqnarray*}
 \lambda & \not \in & g \circ r= g\circ u= 
(f \circ p^\omega) \circ u=f \circ (p^\omega u)\, ,\\
\lambda & \in & g \circ s= g\circ v=(f \circ p^\omega) \circ v=
f \circ (p^\omega v)\, ,\\
\lambda &\in& g\circ s= (f\circ p)\circ s= f\circ w \text{ and }\\
\lambda &\not\in & f \circ s = f\circ v\, . 
\end{eqnarray*}
And we can finish this case in the same manner as the previous ones.
\qed\end{proof}

\end{document}